\documentclass[prl,twocolumn,altaffilletter,numerical,flushbottom,secnumarabic,superscriptaddress,floatfix]{revtex4-2}
\usepackage[utf8]{inputenc}
\usepackage[american]{babel}
\usepackage[T1]{fontenc}
\usepackage[pdftex]{graphicx}  
\usepackage{xcolor}
\usepackage{dcolumn}
\usepackage{comment}
\usepackage{bm}
\usepackage{dsfont}
\usepackage{amsthm}
\usepackage{mathrsfs}
\usepackage{bbold}
\usepackage{braket}
\usepackage{wrapfig}
\usepackage{mathtools}
\usepackage{float}
\usepackage{orcidlink}
\usepackage{physics}
\usepackage[normalem]{ulem}
\usepackage{etoolbox}
\usepackage{comment}

\usepackage{hyperref}
\hypersetup{colorlinks,bookmarksopen,bookmarksnumbered,citecolor=red,linkcolor=red,urlcolor=red}

\newcommand{\Id}{\mathbb{1}}

\newcommand{\MNL}{M_2^{\mathrm{NL}}}

\newcommand{\I}{\mathbb I}

\newtheorem{theorem}{Theorem}
\newtheorem{lemma}[theorem]{Lemma}

\theoremstyle{definition}

\theoremstyle{remark}

\newcommand{\RR}{\mathbb{R}}
\newcommand{\pf}{\operatorname{pf}}

\begin{document}

\title{Nonlocal nonstabilizerness in free fermion models}

\author{Mario Collura\,\orcidlink{0000-0003-2615-8140}}
\email{mcolura@sissa.it}
\affiliation{International School for Advanced Studies (SISSA), via Bonomea 265, 34136 Trieste, Italy}
\affiliation{INFN Sezione di Trieste, 34136 Trieste, Italy}
\author{Benjamin B\'eri\,\orcidlink{0000-0001-9933-9108}}
\email{bfb26@cam.ac.uk}
\affiliation{T.C.M. Group, Cavendish Laboratory, University of Cambridge, J.J. Thomson Avenue, Cambridge, CB3 0HE, UK\looseness=-1}
\affiliation{DAMTP, University of Cambridge, Wilberforce Road, Cambridge, CB3 0WA, UK}
\author{Emanuele Tirrito\,\orcidlink{0000-0001-7067-1203}}
\email{emanuele.tirrito@epfl.ch}
\affiliation{Institute of Physics, Ecole Polytechnique Fédérale de Lausanne (EPFL), CH-1015 Lausanne, Switzerland}
\affiliation{Center for Quantum Science and Engineering, Ecole Polytechnique Fédérale de Lausanne (EPFL), CH-1015 Lausanne, Switzerland}

\begin{abstract}
Nonlocal magic quantifies the irreducible nonstabilizerness of a bipartite quantum state after optimizing over local basis changes. We study nonlocal magic for pure fermionic Gaussian states, and derive a simple closed-form entanglement spectrum bound in terms of the singular values of the subsystem-restricted covariance matrix. We benchmark our result against simulated annealing over local Gaussian unitary transformations, which supports optimality along the full local Gaussian orbit. For states drawn from the Gaussian Haar ensemble, we show that the average nonlocal magic is extensive and determine its thermodynamic limit using random matrix theory for the appropriate circular unitary ensemble. 
We also study Gaussian ground states, focusing on the Kitaev chain, and find that nonlocal magic is suppressed deep in both trivial and topological phases and peaks near the critical points.
Finally, we investigate Gaussian evolution via random circuits and in quenches with the XY chain. For random circuits, we find that nonlocal magic grows diffusively, while in the XY chain the XX limit reveals a striking separation between nonlocal magic and entanglement. 
\end{abstract}

\maketitle

\paragraph{Introduction.---}
It is increasingly recognized that entanglement---one of the essential resources for quantum processing and quantum technology~\cite{shor_1997,RevModPhys.82.277,cirac2012goals,RevModPhys.76.1037}---by itself is insufficient for quantum advantage~\cite{PhysRevLett.91.147902,PhysRevA.75.012337,daley2022practical,bravyi2018quantum,harrow2017quantum,preskill2018quantum,preskill2012quantum,RevModPhys.95.035001,RevModPhys.91.025001}. 
For instance, one can obtain highly entangled so-called stabilizer states by Clifford circuits \cite{nielsen2010quantum} that can be efficiently simulated classically \cite{PhysRevA.57.127,PhysRevA.70.052328,PhysRevLett.112.240501}. Stabilizer states also provide important building blocks in the construction of quantum-error-correcting codes \cite{PhysRevLett.102.110502,PhysRevA.55.900,PhysRevA.54.4741,PhysRevLett.84.2525,PhysRevX.6.031006,PhysRevB.103.104306,PhysRevA.65.012308,PhysRevLett.77.198,PhysRevA.54.1862,yi2024complexity}.

Building on this understanding, a new metric, referred to as ``magic,'' or nonstabilizerness, was introduced to quantify the amount of non-Clifford resources necessary to prepare a given quantum state, providing a more nuanced measure of a state’s quantumness \cite{PhysRevA.71.022316,PhysRevA.86.052329,campbell2017roads,PhysRevLett.116.250501,PhysRevX.6.021043,garcia2017geometry,veitch2014resource, gottesman1998heisenberg,gottesman1998fault}.
Magic has since been connected to a number of fundamentally quantum phenomena \cite{PhysRevA.83.032317,PhysRevX.6.021043,beverland2020lower,PhysRevLett.128.050402,PhysRevA.107.022429}.

The theory of nonstabilizerness in many-body settings has seen rapid development recently \cite{PhysRevB.103.075145,ellison2021symmetry,PhysRevB.106.125130,sarkar2020characterization,PRXQuantum.3.020333}, and much effort has been put into finding corresponding computable measures of magic \cite{PRXQuantum.3.020333,PhysRevA.106.042426,PhysRevB.107.035148,sonya2025nonstabilizerness}. An emerging theme is that nonstabilizerness could be a useful tool to characterize
many-body states and phases of matter \cite{PRXQuantum.4.040317,PhysRevB.111.L081102,1tyr-rlbb,ylsz-dm3y,4hpw-6mq3,PhysRevLett.131.180401}, as
well as quantum dynamics \cite{turkeshi2025magic,p7xt-s9nz,1jzy-sk9r,xfp5-hhs4,tirrito2025universal,PRXQuantum.6.010345,PhysRevA.108.042408,PhysRevA.107.022429,PhysRevA.109.032403,PhysRevB.111.054301,10.21468/SciPostPhys.19.6.159,sierant_2026,braccia_2026,PRXQuantum.5.030332,PhysRevA.108.042407}. Several works have
also investigated the interplay between nonstabilizerness and other physical properties, most notably entanglement~
\cite{PhysRevA.109.L040401,PhysRevA.108.042408,PRXQuantum.6.020324,PhysRevResearch.6.L042030,PhysRevB.110.045101,tarabunga2025magic,PRXQuantum.5.030332}.

The relation to entanglement raises the fundamental question: how much of magic is \emph{nonlocal}? For a bipartite pure state $|\psi\rangle_{AB}$, nonlocal magic is the residual nonstabilizerness that remains after minimizing magic through local basis changes~\cite{cao_2024, cao_2025, qian_2025}.
This quantity isolates the irreducible component of magic that cannot be removed within the two subsystems separately, and therefore captures genuinely nonlocality in magic. However, the optimization is generically formidable: magic measures are nonlinear, the local-unitary orbit is exponentially large, and even evaluating magic itself is typically hard beyond small systems. As a result, explicit results for nonlocal magic are scarce, especially in many-body settings.
 
In this work, we study nonlocal magic for pure fermionic Gaussian states $|\psi\rangle_{AB}$ \cite{fradkin2013field,giuliani2008quantum} in terms of the $\alpha$-stabilizer R\'enyi entropy ($\alpha$-SRE)~\cite{PhysRevLett.128.050402,haug2023stabilizer,PhysRevA.110.L040403}. 
We show that nonlocal magic admits closed-form upper bound $M_{\alpha,>}^\text{NL}(\psi)$---which we believe to be tight
---in terms of covariance-matrix invariants related to the entanglement spectrum.
This is analogous to entanglement-spectrum bounds on nonlocal $2$-SRE for general pure states~\cite{cao_2025}, but while these are intricate, our focus on Gaussian states allows us to obtain simple and efficiently computable expressions, for all $\alpha$.  Concretely, $M_{\alpha,>}^\text{NL}(\psi)$ reduces to a sum of independent mode contributions labeled by the singular values of the subsystem covariance matrix \(i\Gamma_{AA}\). 

We benchmark our analytical result against simulated annealing over local Gaussian unitary transformations, and find numerical evidence suggesting that it captures the optimum along the local Gaussian orbit. We then study the Gaussian Haar ensemble and show that the average $M_{\alpha,>}^\text{NL}(\psi)$ is extensive, and provide a closed-form expression in the thermodynamic limit. Finally, we study ground states and dynamics of local free-fermion systems: for the ground state of the Kitaev chain, we find that$M_{\alpha,>}^\text{NL}(\psi)$ is strongly suppressed deep in both trivial and topological phases, while it is enhanced near the critical points; for brick-wall matchgate circuits we observe diffusive growth of $M_{\alpha,>}^\text{NL}(\psi)$; for quenches in the XY chain we show that the while the dynamics of entanglement and $M_{\alpha,>}^\text{NL}(\psi)$ are generically similar, they become qualitatively distinct at the U($1$)-symmetric XX limit.

\paragraph{Nonstabilizernees for free fermions.--}
We consider a quantum system of $L$ qubits, with total Hilbert space dimension $D=2^L$. Onsite Pauli operators are  $P_i\in\lbrace I,X,Y,Z \rbrace$ for $i=1,2,\cdots L$, while $\textbf{P}\in \left\lbrace I, X, Y, Z \right\rbrace^{\otimes L} $ represents their generic tensor product  (i.e. a Pauli string). Majorana operators are defined using the  Jordan-Wigner mapping: $\gamma_{2i-1} = Z_1 \cdots Z_{i-1} X_i I_{i+1} \cdots I_{L}$ and $\gamma_{2i} = Z_1 \cdots Z_{i-1} Y_i I_{i+1} \cdots I_{L}$. The $2L$ Majorana operators are Hermitian and satisfy the canonical anti-commutation relations $\lbrace \gamma_{\mu},\gamma_{\nu}\rbrace=2\delta_{\mu \nu}$. Majorana monomials are in one-to-one correspondence with Pauli strings, forming a complete orthogonal basis as well, since $\text{Tr}[(\gamma^{\textbf{x}'})^{\dagger}\gamma^{\textbf{x}}]=D \delta_{\textbf{x}',\textbf{x}}$, where $\textbf{x}\in \lbrace 0,L \rbrace^{2L}$. A unitary operator $U$ is Gaussian if and only if it acts on the Majorana operators by rotationg the vector of $\gamma_{\mu}$ as $U^{\dagger} \gamma_{\mu} U= \sum_{\nu} O_{\mu \nu} \gamma_{\nu}$ with $O \in SO(2L)$.
Gaussian states $\rho_{\Gamma}$ are fully characterized by the real skew-symmetric covariance matrix~\cite{surace2022fermionic,hackl_2021}
\begin{equation}
    \Gamma_{\mu \nu} =-\frac{i}{2} \mathrm{Tr}([\gamma_{\mu},\gamma_{\nu}]\rho_{\Gamma})
\end{equation}
which contains all two-point Majorana correlation functions. When applying a Gaussian unitary to a state $\rho_{\Gamma}$, the covariance matrix rotates accordingly as $\Gamma \to O \Gamma O^{T}$. 

Wick’s theorem enables the calculation of any correlation function in terms of two-point correlators and can be formulated as follows $\mathrm{Tr} \left( \rho \gamma^{\boldsymbol{x}} \right)=i^{\frac{|\boldsymbol{x}|}{2}} \mathrm{Pf}\left[\Gamma|_{\boldsymbol{x}} \right]$,
where $\Gamma|_{\boldsymbol{x}}$ denotes the square sub-matrix of $\Gamma$ that includes all rows and columns associated with indices equal to $1$ in $\boldsymbol{x}$, and $|\boldsymbol{x}$| represents the total count of such indices.

Now we can write the SRE \cite{PhysRevLett.128.050402} in terms of minors of the covariance matrix. For a generic state $\rho$, let us introduce the probability density of Pauli strings, $\pi_{\rho}(\boldsymbol{P})=|\mathrm{Tr}(\rho \boldsymbol{P})|^2/2^L$, using which $\alpha$-SRE is
\begin{equation}
    M_{\alpha}(\rho)
    =
    \frac{1}{1-\alpha}
    \log \sum_{\boldsymbol{P}}
    \pi_{\rho}(\boldsymbol{P})^{\alpha}
    - \log D .
\end{equation}
Thus, up to the additive constant $-\log D$, the SRE coincides with the  $\alpha$-Rényi entropy of $\pi_{\rho}(\boldsymbol{P})$. 
For fermionic Gaussian states, $\pi_{\rho}(\boldsymbol{P})$ can be evaluated using the covariance matrix: Majorana monomials $\gamma^{\boldsymbol{x}}$ are in one-to-one correspondence with Pauli strings, so that $\pi_{\rho}(\boldsymbol{P})$ can equivalently be represented as a distribution $\pi_{\Gamma}(\boldsymbol{x})$ over Majorana strings. 
Using Wick's theorem and the identity $\mathrm{Pf}(A)^2=\det A$, one obtains $\pi_{\Gamma}(\boldsymbol{x})=\frac{\det[\Gamma|_{\boldsymbol{x}}]}
    {\det[\mathbb{I}+\Gamma]}$.
Here $\Gamma|_{\boldsymbol{x}}$ denotes the principal submatrix of $\Gamma$ selected by the support of $\boldsymbol{x}$, with the convention that odd-dimensional minors vanish. The denominator follows from the principal-minor expansion 
$\det(\mathbb{I}+\Gamma)=\sum_{\boldsymbol{x}}\det[\Gamma|_{\boldsymbol{x}}]$ and normalizes the distribution; equivalently, it is $D\,\mathrm{Tr}(\rho_{\Gamma}^{2})$, and reduces to $D$ for pure Gaussian states. Therefore, for a Gaussian state the SRE \cite{collura2026non} is 
\begin{equation}\label{eq:MGam}
    M_{\alpha}(\Gamma)=\frac{1}{1-\alpha} \log \sum_{\boldsymbol{x}} \left( \frac{\det[\Gamma|_{\boldsymbol{x}}]}{\det[\I+\Gamma]} \right)^{\alpha} -\log D .
\end{equation}

\paragraph{Nonlocal nonstabilizerness for free fermions.--} We now turn to the bipartite setting. For a pure state $|\psi\rangle$, the nonlocal magic, in terms of the SRE shared between subsystems $A$ and $B$, is~\cite{cao_2024, cao_2025, qian_2025}
\begin{equation}
     M_\alpha^{\rm NL}(\psi)= \min_{U_A\otimes U_B}
M_\alpha\!\left[(U_A\otimes U_B)|\psi\rangle\right],
\end{equation}
where the unitary $U_{X}$ acts on subsystem $X$.
For generic many-body states, finding $M_\alpha^{\rm NL}(\psi)$ is highly nontrivial. For pure fermionic Gaussian states $\ket{\psi_\Gamma}$, however, it becomes more tractable. We take $A$ to have $m$ and $B$ to have $n$ qubits (i.e., $2n$ and $2m$ Majorana fermions and $m+n=L$) and consider the block decomposition
\begin{equation}
\Gamma=
\begin{pmatrix}
\Gamma_{AA} & \Gamma_{AB}\\
-\Gamma_{AB}^T & \Gamma_{BB}
\end{pmatrix},
\qquad
\Gamma^T=-\Gamma,\qquad \Gamma^2=-\mathbb{I}.
\end{equation}
Local Gaussian unitaries act as $\Gamma\mapsto (O_A\oplus O_B)\Gamma(O_A\oplus O_B)^T $ with $O_A\in SO(2m)$ and $O_B\in SO(2n)$, and there exist 
$O^*_{A/B}$ that rotate $\Gamma$ to the bipartite canonical form (assuming $m \leq n$)~\cite{botero_2004,SchuchBauer_prb19,Bejan_MCDT}
\begin{equation} \label{eq:gamma_can}
\Gamma_{\mathrm{can}}:=(O_A^\star\oplus O_B^\star)\Gamma(O_A^\star\oplus O_B^\star)^T
=
\bigoplus_{k=1}^{m}\Gamma_k
\oplus
\bigoplus_{\ell=1}^{n-m}J_\ell,
\end{equation}
where
\begin{equation} 
\Gamma_{k}=
\begin{pmatrix}
\nu_k J_\ell & \mu_k Z\\
-\mu_k Z & \nu_k J_\ell
\end{pmatrix},
\qquad
\mu^2_k=1-\nu_k^2,
\end{equation}
with $ J_\ell=i Y$. The numbers \(\{\nu_k\}=\mathrm{spec}_+(i\Gamma_{AA})\) set the entanglement spectrum~\cite{Fidkowski_PhysRevLett.104.130502}. The state $\ket{\psi_{\Gamma_{\mathrm{can}}}}$ consists of independent elementary \(A\)-\(B\) entangled pairs plus trivial local modes.  As we argue below, a useful upper bound $M_{\alpha,>}^\text{NL}(\psi_\Gamma)$ on $M_{\alpha}^\text{NL}(\psi_\Gamma)$ is $M_{\alpha}(\Gamma_{\mathrm{can}})$, hence, by Eq.~\eqref{eq:MGam},   
\begin{equation} \label{eq:MNLG_main}
M_{\alpha,>}^\text{NL}(\psi_\Gamma) =  \frac{1}{1-\alpha}\sum_k \log\!\left( (1+\nu_k^{2\alpha}+\mu_k^{2\alpha})/2\right).
\end{equation}
This is one of our key results. While entanglement-spectrum bounds for
$M_{2}^\text{NL}$ are generally involved~\cite{cao_2025}, focusing on
Gaussian states allows us to obtain simple expressions for $M_{\alpha,>}^\text{NL}(\psi_\Gamma)$ for all $\alpha$. In this case, the bound reduces to an additive mode-resolved expression fixed solely
by the single-particle entanglement spectrum \(\{\nu_k\}\), 
which vanishes both for locally pure modes
$(\nu_k=1)$ and for maximally entangled stabilizer-like Gaussian pairs $(\nu_k=0)$. Thus $M_{\alpha,>}^{\rm NL}$
probes only intermediate modes, $0<\nu_k<1$, isolating the nonstabilizer component of Gaussian correlations across the cut.

To interpret Eq.~\eqref{eq:MNLG_main}, consider the following preparation
of $\ket{\psi_{\Gamma_{\mathrm{can}}}}$: start with $\otimes_j \ket{0}_j$, corresponding to $\Gamma_\text{vac} = \oplus_{k=1}^{L} i Y$, i.e., $\nu_k=1$ for all $k$. To entangle the $k^\text{th}$ pair, initially in $\ket{0_{k_A}}\ket{0_{k_B}}$, we apply $\exp(i\phi_k X_{k_A}X_{k_B})$ which sets $\nu_k=\cos(2\phi_k)$. After entangling the $m$ pairs, we have (up to the decoupled $\otimes_{j=1}^{n-m}\ket{0}_j$) a tensor product of $\ket{\overline{\psi}_k}=\cos(\phi_k)\ket{\overline{0}_k}+i\sin(\phi_k)\ket{\overline{1}_k}$ where $\ket{\overline{j}_k}=\ket{j_{k_A}}\ket{j_{k_B}}$ ($j\in\{0,1\})$. For $\phi\neq n\pi/4$ (with $n\in \mathbb{Z}$) the gates introduce magic and in Eq.~\eqref{eq:MNLG_main} we recognize the $\alpha$-SRE from the product of the $\ket{\overline{\psi}_k}$, i.e., the sum of the corresponding single-qubit $\alpha$-SRE~\cite{PhysRevLett.128.050402}.

Since $\ket{0_{k_A}}\ket{0_{k_B}}$ is a stabilizer state, the magic in $\ket{\psi_{\Gamma_{\mathrm{can}}}}$ is purely nonlocal: it came to be solely via entanglement. 
Furthermore, being a two-qubit repetition code logical state, each $\ket{\overline{\psi}_k}$ has $\pi_{\rho}(\boldsymbol{P})\neq 0$ for the smallest number (i.e., six) of $\boldsymbol{P}$ compatible with entanglement and nonstablizerness. This suggests that local  unitaries cannot localize $\pi_{\rho}(\boldsymbol{P})$, hence reduce $\alpha$-SRE, further, thus
$M_{\alpha,>}^\text{NL}(\psi_\Gamma)= M_{\alpha}^\text{NL}(\psi_\Gamma)$
for all $\alpha>1$. While we can guarantee only $M_{\alpha,>}^\text{NL}(\psi_\Gamma)\geq M_{\alpha,\text{G}}^\text{NL}(\psi_\Gamma) \geq M_{\alpha}^\text{NL}(\psi_\Gamma)$, where $M_{\alpha,\text{G}}^\text{NL}(\psi_\Gamma)$ results from minimizing only via Gaussian unitaries, the upper bound $M_{\alpha,>}^\text{NL}(\psi_\Gamma)$ may thus be tight.

In the Supplemental Material~\cite{supp_material}, we show that $M_{\alpha,>}^\text{NL}(\psi_\Gamma)$ is a local minimum of $M_{\alpha,\text{G}}^\text{NL}(\psi_\Gamma)$.

\begin{figure}[t!]
\includegraphics[width=0.24\textwidth]{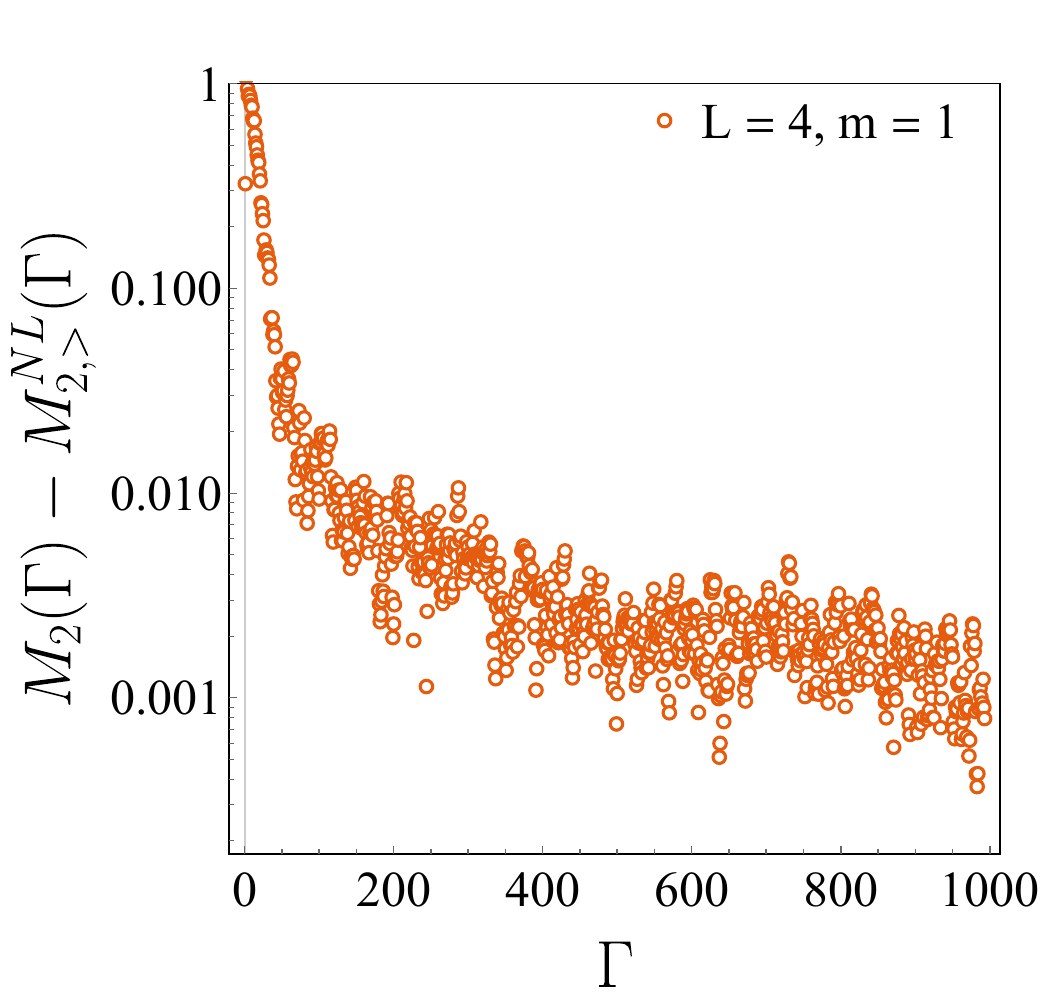}\includegraphics[width=0.24\textwidth]{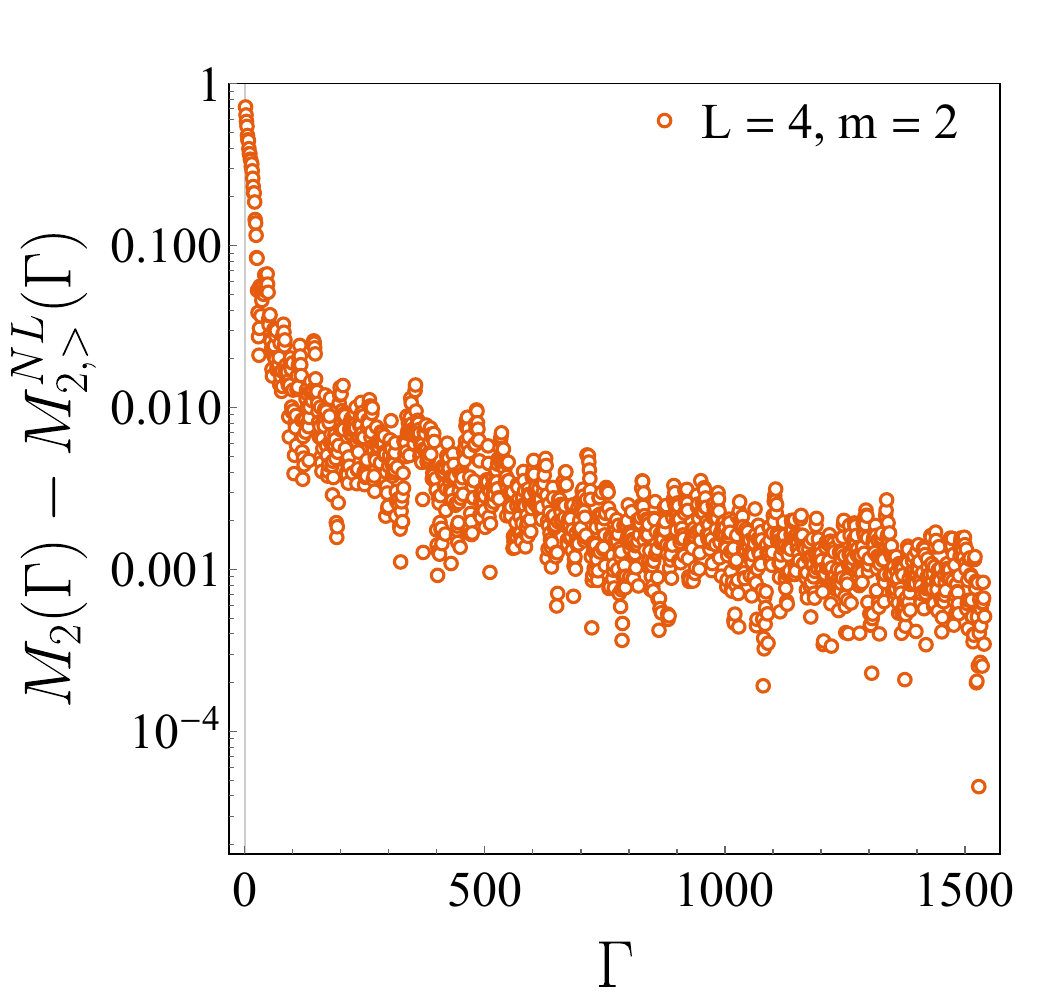}
\caption{ \textbf{Simulated-annealing benchmark of the analytical nonlocal-magic bound.} For random pure Gaussian states with $2L=8$ Majorana modes, we minimize $M_2(\Gamma)$ over the local Gaussian orbit and plot $M_2(\Gamma)-M^{\rm NL}_{2,>}(\psi_\Gamma)$ along the annealing trajectory. 
The two panels correspond to bipartitions with $|A|=2m=2$ (left) and $|A|=2m=4$ (right). The trajectories approach the analytical value without crossing it, providing numerical evidence that $M^{\rm NL}_{2,>}(\psi_\Gamma)$ attains the Gaussian local minimum.  The annealing schedule uses $d\beta=0.5$.
} \label{fig:annealing}
\end{figure}

\paragraph{Numerical benchmark via simulated annealing.--}
We now compare $M_{\alpha,>}^\text{NL}(\psi_\Gamma)$ and $M_{2,\text{G}}^\text{NL}(\psi_\Gamma)$ numerically: we sample random pure Gaussian states and explore their local orthogonal orbit through simulated annealing. 

We generate states through their covariance matrix $\Gamma_{O} = O \Gamma_\text{vac} O^{T}$ using Haar-random  $O \in SO(2L)$. For every state, we first evaluate $M_{2,>}^\text{NL}(\psi_\Gamma)$ from Eq.~\eqref{eq:MNLG_main} and then compare it against $M_{2}(\Gamma)$ from Eq.~\eqref{eq:MGam}.
Since this exact computation grows exponentially with system size, the benchmark could only be performed for relatively small systems.

To test whether  $M_{2,>}^\text{NL}(\psi_\Gamma)$ indeed gives the minimum over all local Gaussian basis, we perform simulated annealing: at each Monte Carlo step, we propose an update of the form
$\Gamma \longmapsto (O_A\oplus O_B)\,\Gamma\,(O_A\oplus O_B)^T$ where, to produce small updates, we generate $O_{A/B}$ as exponentials of random antisymmetric matrices,
$O_{A/B} = e^{\epsilon K_{A/B}}$, with each $K=(Q-Q^T)/2$ obtained from a Gaussian random matrix $Q$.
The proposed move is accepted with the standard Metropolis rule,
$p=\min\left\{1,\exp\left[-\beta\left(M_2(\Gamma')-M_2(\Gamma)\right)\right]\right\}$,
where $\Gamma'=(O_A\oplus O_B)\Gamma(O_A\oplus O_B)^T$
and $\beta$ is the inverse temperature. As the annealing proceeds, $\beta$ is gradually increased while the typical step size is reduced according to $\epsilon\sim \beta^{-1/2}$ so that the dynamics initially explores the landscape broadly and later focuses on a small neighborhood of the lowest values encountered. 

For a random initial condition, we record the full annealing trajectory and compare the minimum $M_{2}(\Gamma)$ reached numerically with $M_{2,>}^\text{NL}(\psi_\Gamma)$. Repeating this analysis over many independently generated Gaussian states, we consistently found that $M_{2,>}^\text{NL}(\psi_\Gamma)$ provides a \emph{lower} bound to $M_{2}(\Gamma)$ along the entire trajectory, suggestive of $M_{\alpha,>}^\text{NL}(\psi_\Gamma)$ achieving minimization over local Gaussian unitaries, i.e.,  $M_{\alpha,>}^\text{NL}(\psi_\Gamma) = M_{\alpha,\text{G}}^\text{NL}(\psi_\Gamma)$ (see Fig.~\ref{fig:annealing} for two trajectories corresponding to two different partitions).

\begin{figure}[t!]
\includegraphics[width=0.48\textwidth]{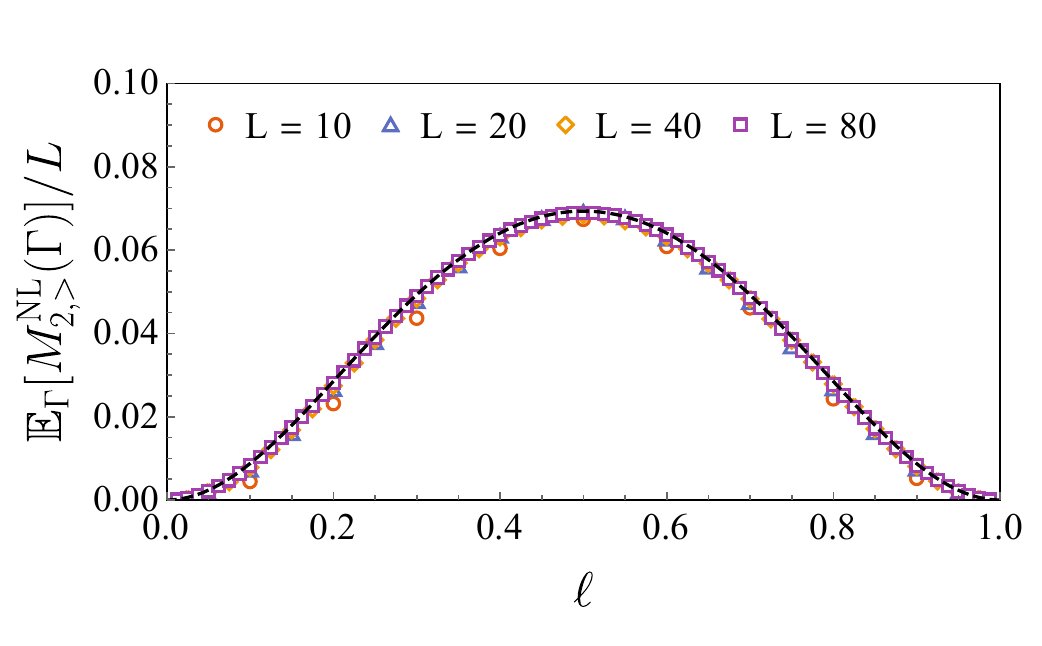}
\caption{\textbf{Random Gaussian ensemble average of the nonlocal magic density.} The quantity $\mathbb{E}_{\Gamma}[M^{\rm NL}_{2,>}(\psi_{\Gamma})]/L$ is shown as a function of the bipartition fraction $\ell=m/L$. The dashed black line denotes the thermodynamic prediction $J(\ell)$ in Eq.~(\ref{eq:m2NL_avg_th}), while symbols are numerical averages of Eq.~(\ref{eq:MNLG_main}) over $100$ Haar-random Gaussian states for different $L$. 
The finite-size data converge to the Jacobi-ensemble result, with a maximum at the symmetric bipartition.
} \label{fig:random_gauss}
\end{figure}
\paragraph{Random Gaussian Ensemble Average.--}
As an application, we compute the average non-local $2$-SRE density over the Gaussian Haar ensemble for a generic bipartition $A{:}B$ with $|A|=2m$, $|B|=2n$ and $m+n=L$. Since Eq.(\ref{eq:MNLG_main}) is additive (i.e., a linear statistic~\cite{BeenRMT}), 
\begin{equation}\label{eq:avg_M2NL}
\frac{\mathbb{E}_{\Gamma}[M_{2,>}^\text{NL}(\psi_\Gamma)]}{L} = 
-\frac{m}{L}\int_{0}^{1} dx \log(1-x+x^2) \rho_{m,n}(x)
\end{equation}
where $\rho_{m,n}(x)$ is the one-point density of $x=\nu^2$. For the Gaussian Haar ensemble, the joint probability density $P_{m,n}(x_1,\dots,x_m)$ of $x_k = \nu_k^2 \in [0,1]$ for $k=1,\dots,m$ is of Jacobi-type~\cite{bianchi_2021,sierant_2026,supp_material} and corresponds to the class DIII~\cite{AltlandZirnbauer} circular ensemble in terms of covariance matrices~\cite{Bejan_MCDT}. While $\rho_{m,n}(x)$ have exact expressions in terms of Jacobi polynomials~\cite{bianchi_2021,jiang_2009}, here we are interested only in the thermodynamic regime  $L\to \infty$ with fixed $\ell=\frac{m}{L}$, where $\rho_{m,n}(x)\to \rho_(x)$ with~\cite{BeenRMT,Dahlhaus_PhysRevB.82.014536,supp_material}
\begin{equation}\label{eq:rho1}
\rho(x)
=
\frac{1}{2\pi\ell}\,
\frac{\sqrt{(\tilde x-x)/x}}{1-x}\,
\mathbf{1}_{(0,\,\tilde x)}(x),\quad \tilde x = 4\ell (1-\ell). 
\end{equation}
Using this in Eq.~(\ref{eq:avg_M2NL}),  
$\mathcal{J}(\ell) \equiv \lim_{L\to\infty} \frac{1}{L} \mathbb{E}_{\Gamma}[M_{2,>}^\text{NL}(\psi_\Gamma)]$ evaluates to~\cite{supp_material}
\begin{eqnarray}\label{eq:m2NL_avg_th}
\mathcal{J}(\ell) & = & \Re 
\Bigg\{
2\sqrt{1-\tilde x}\log\left[ \frac{\sqrt{1-\tilde x}+\sqrt{1-\tilde x e^{i\pi/3}}}{\sqrt{1-\tilde x}+1}\right]\nonumber \\
& - &
2\log\left[ \frac{1+\sqrt{1-\tilde x e^{i\pi/3}}}{2}\right]
\Bigg\}.
\end{eqnarray}
We find that $\mathcal{J}(\ell)$ is maximal at the symmetric bi-partition, i.e. $\mathcal{J}(1/2) = \log(8-4\sqrt{3}) \simeq 0.0693365$; the non-local magic of the Gaussian ensemble, although being extensive, contributes only a small fraction of the total non-stabiliserness (for a Gaussian random state $\mathbb{E}_{\Gamma}[M_2 (\psi_{\Gamma})] \simeq cL$ with $c\simeq 1$~\cite{sierant_2026}).
In Fig.~\ref{fig:random_gauss}, we show the convergence of the averaged non-local Gaussian magic towards $\mathcal{J}(\ell)$ as $L$ increases.

\begin{figure}[t]
\includegraphics[width=0.24\textwidth]{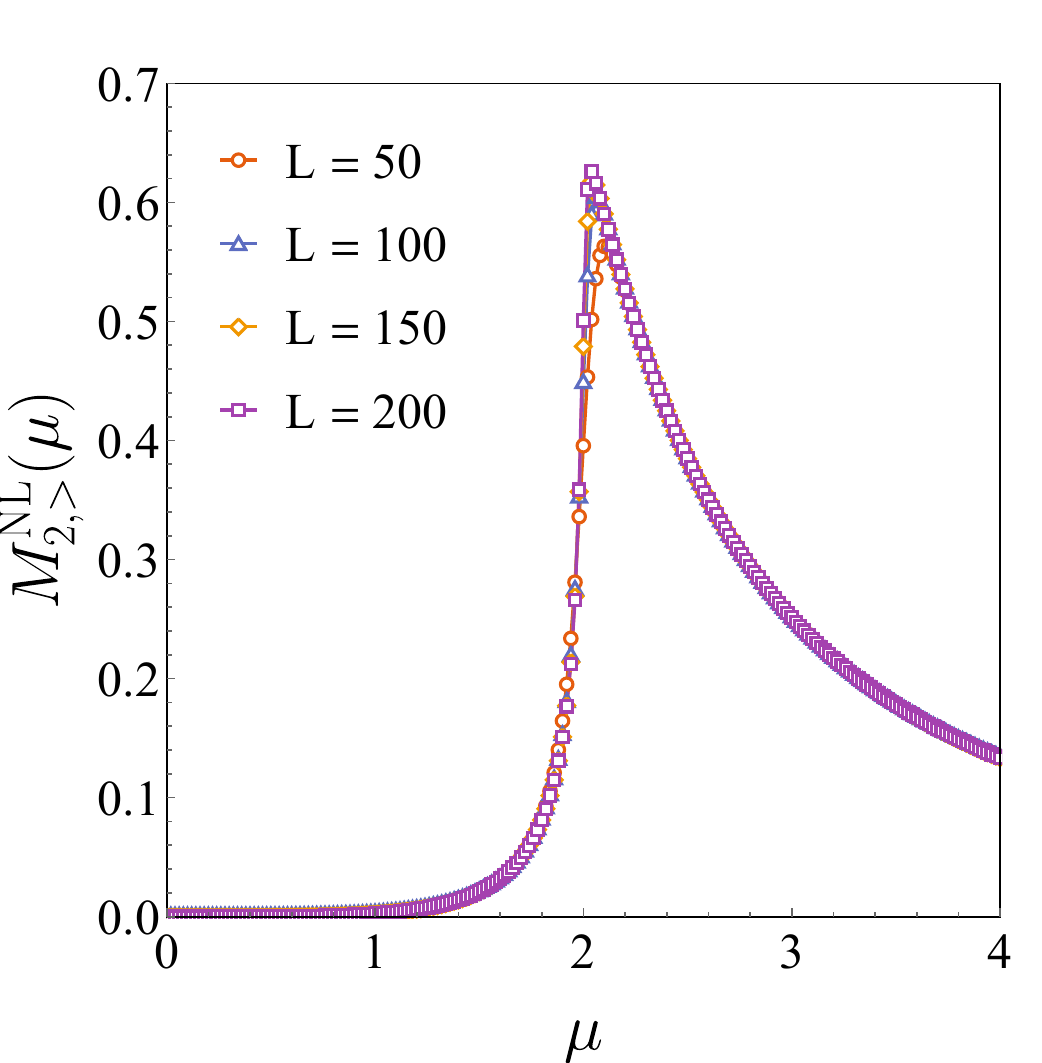}\includegraphics[width=0.24\textwidth]{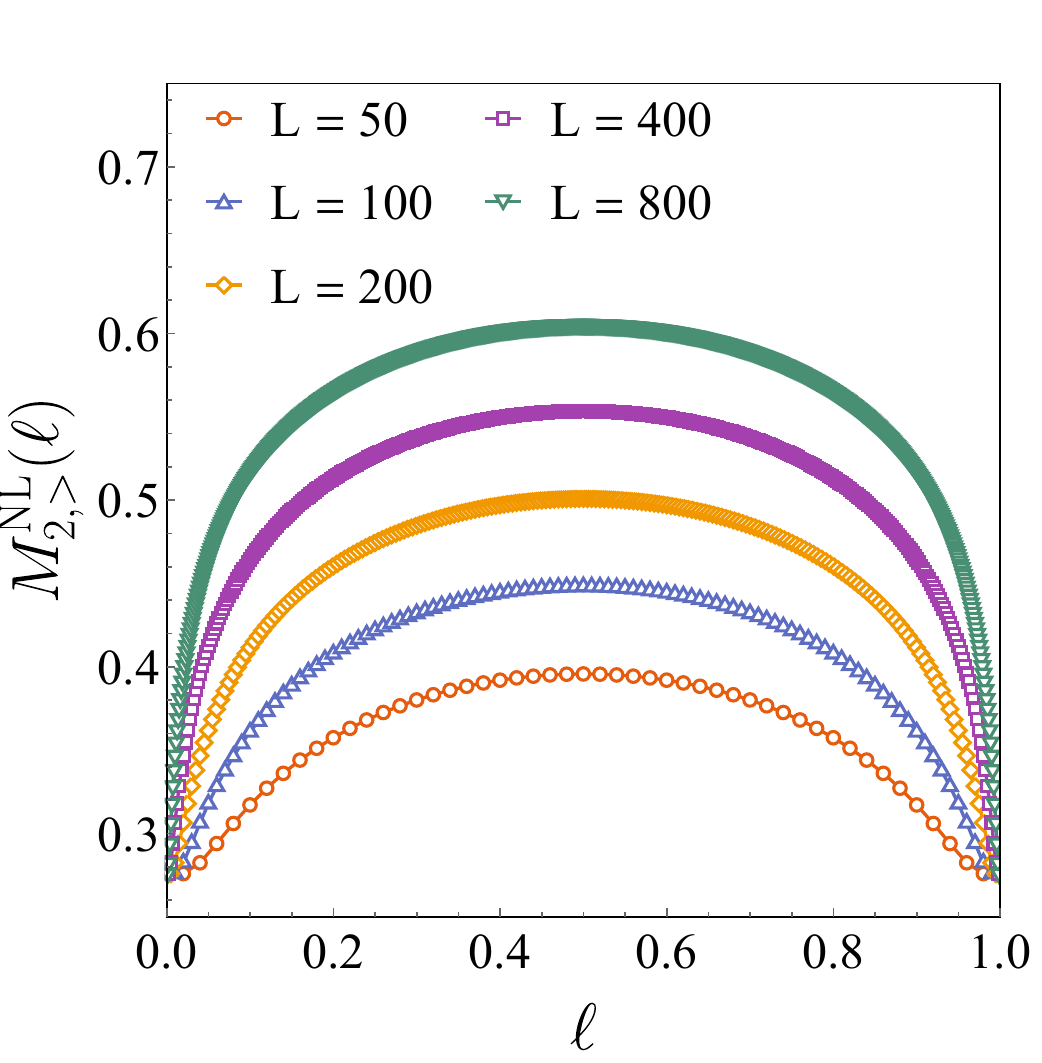}
    \caption{\textbf{Ground-state nonlocal magic in the Kitaev chain.}
    (Left panel) Nonlocal Gaussian magic $M^{\rm NL}_{2,>}(\psi_\Gamma)$ for the ground state of the Kitaev chain with $t=\Delta=1$, shown as a function of $\mu$ for different system sizes and a symmetric bipartition. The peak sharpens near the critical point $\mu=2$. (Right panel) $M^{\rm NL}_{2,>}(\psi_\Gamma)$ at criticality, $\mu=2$, as a function of the bipartition fraction $\ell=m/L$. The growth with system size is consistent with logarithmic critical scaling.
    }
    \label{fig:kitaev}
\end{figure}

\paragraph{Hamiltonian ground states.--}
As a next application, we consider ground states of the one-dimensional Kitaev Hamiltonian $H_{\rm K}=
\sum_j \left \lbrace \left(-t c_j^\dagger c_{j+1}+\Delta  c_j c_{j+1}+{\rm h.c.}\right) -\mu\left(c_j^\dagger c_j-\frac{1}{2}\right) \right \rbrace$. 
For \(\Delta=t\), the model undergoes a topological transition at \(|\mu|=2t\). In right panel of Fig.~\ref{fig:kitaev}, we show \(M^{\rm NL}_{2,>}(\psi_\Gamma)\) 
for \(t=\Delta=1\), as a function of \(\mu\), for increasing system sizes and a symmetric bipartition.
\(M^{\rm NL}_{2,>}(\psi_\Gamma)\) is suppressed deep in both gapped phases (the asymmetry around $\mu=2$ reflects that of the stabilizer limits, $\mu=0$ and $\mu\to\infty$ in the topological and trivial phases, respectively).
\(M^{\rm NL}_{2,>}(\psi_\Gamma)\), however, develops a pronounced peak near the critical point \(\mu=2\). Upon increasing \(L\), the peak sharpens, indicating 
that \(M^{\rm NL}_{2,>}(\psi_\Gamma)\) is sensitive to the critical reorganization of the Gaussian correlations across the bipartition.

Away from criticality, only a finite number of such modes contributes appreciably, leading to a small value of the nonlocal magic. At criticality, instead, the
scale-invariant correlations generate a broad crossover in the spectrum of \(i\Gamma_{AA}\), producing an increasing number of intermediate modes.

This critical enhancement is further resolved in the left panel of Fig.~\ref{fig:kitaev}, where we fix \(\mu=2\) and plot $M^{\rm NL}_{2,>}$ as a function of the bipartition
fraction \(\ell=m/L\). The curves are maximal close to the symmetric cut and increase with system size. At fixed \(\ell\), the growth is consistent with a logarithmic dependence on the subsystem size, reflecting the logarithmic broadening of the Gaussian entanglement spectrum at the critical point.

\begin{figure}[t!]
\includegraphics[width=0.24\textwidth]{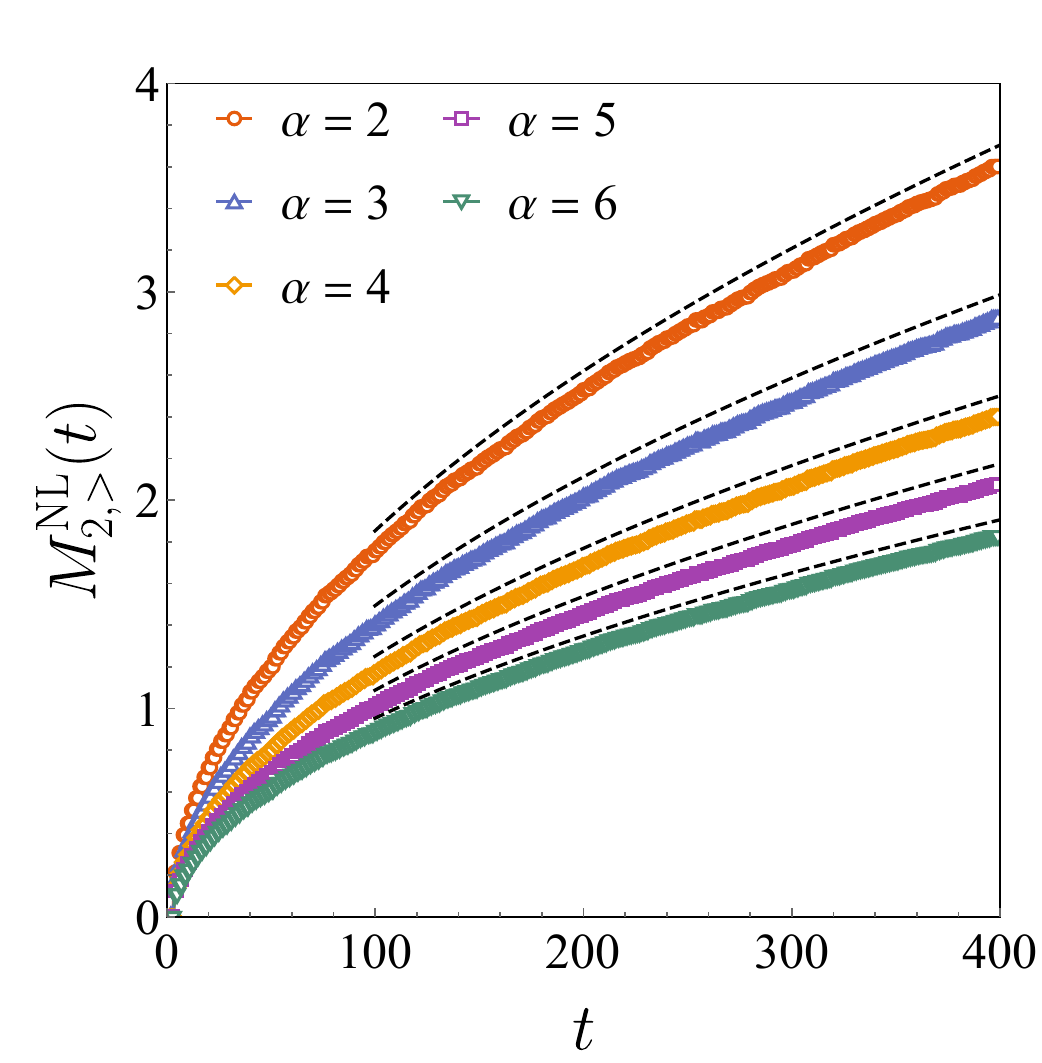}\includegraphics[width=0.24\textwidth]{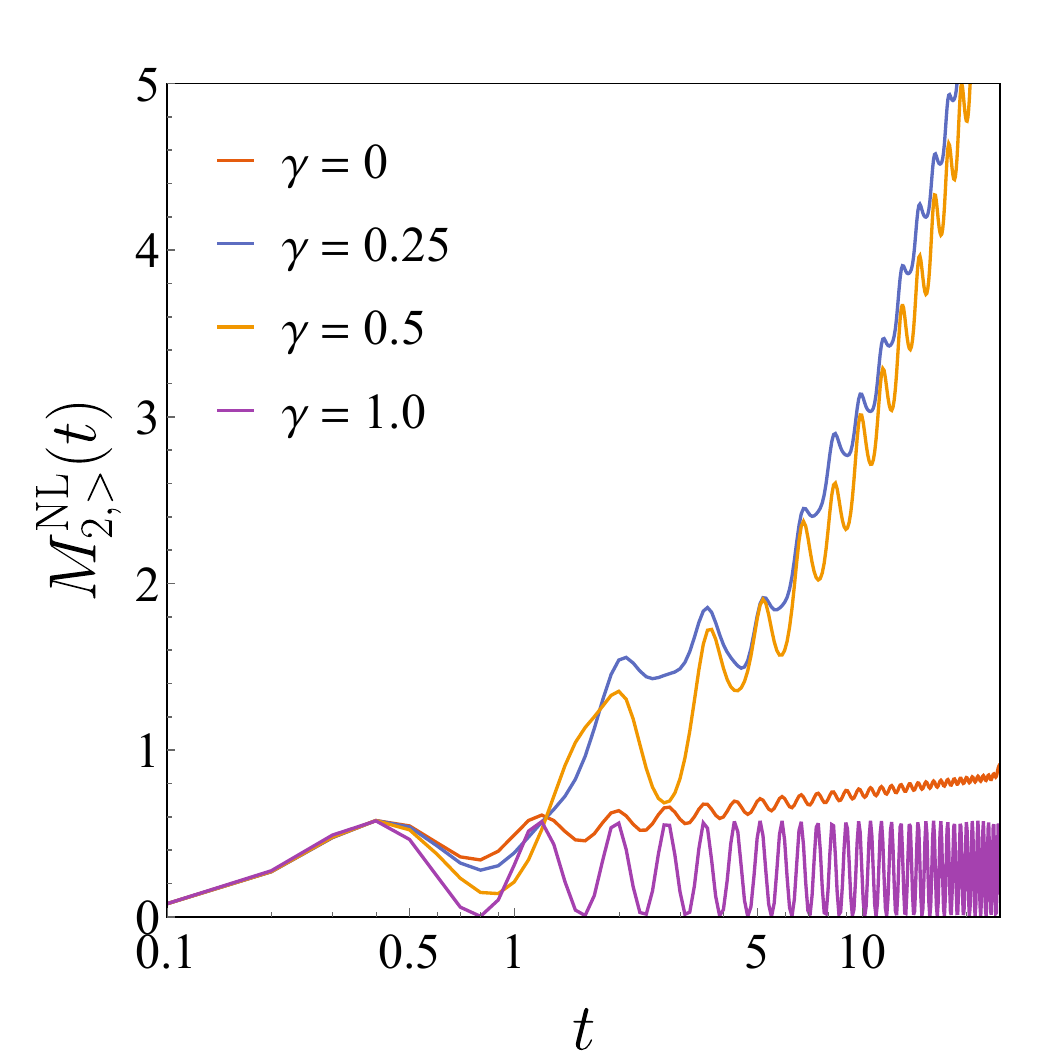}
\caption{
\textbf{Gaussian dynamics of nonlocal magic.} (Left panel) In random Gaussian brick-wall circuits, $M_{\alpha,>}^{\rm NL}(\psi_{\Gamma(t)})$ shows a diffusive behavior for all Rényi indices, with the black-dashed lines highlighting the $t^{1/2}$ scaling.
(Right panel) After a N\'eel-state quench in the XY chain, $M_{2,>}^{\rm NL}(\psi_{\Gamma(t)})$ grows logarithmically at the $U(1)$-symmetric XX point $\gamma=0$, while it grows linearly for $\gamma\neq0$, except in the limit $\gamma=1$.
}
\label{fig:XYquench}
\end{figure}

\paragraph{Gaussian dynamics: random circuits and Hamiltonian quenches.---}
We finally study the real-time generation of nonlocal magic in Gaussian dynamics. We first consider a brick-wall circuit of random two-site Gaussian gates initialized from a product state with $L=200$. As shown in left panel of Fig.~\ref{fig:XYquench},  $M_{2,>}^{\rm NL}$  grows sub-ballistically for all Rényi indices considered, with $M_{\alpha,>}^{\rm NL}(\psi_{\Gamma(t)})\sim t^{1/2}$.

We then study Hamiltonian dynamics, focusing on the XY chain $H_{\mathrm{XY}}
=
\sum_{j=1}^L
\left(
c_j^\dagger c_{j+1}
+
\gamma c_j^\dagger c_{j+1}^\dagger
+
h.c
\right)
$,
following a quench from the N\'eel state. We evolve up to time $t_{\max}=L/4$ to avoid boundary effects and compute $M_{2,>}^{\rm NL}$ across the central bipartition. Fig.~\ref{fig:XYquench}   (right panel) shows that the growth of $M_{2,>}^{\rm NL}$ depends sharply on  $\gamma$.

At the U($1$)-symmetric XX point ($\gamma=0$), $M_{2,>}^{\rm NL}$ grows only logarithmically in time, while for $\gamma\neq 0$ (but away from the $\gamma=1$ limit, i.e., the $t=\Delta=-1$ limit of the Kitaev chain) the growth becomes linear. This reveals a dynamical separation between entanglement (growing linearly for all $\gamma\neq 1$) and nonlocal magic, and suggests that charge conservation strongly suppresses the buildup of irreducible Gaussian nonstabilizerness. 

The physical origin of this effect is in the emergence, after the quench, of a linearly expanding plateau in the bipartite spectrum at $\nu\simeq 0$. This plateau drives the usual ballistic growth of entanglement, but, being near-stabilizer, it contributes negligibly to $M_{\alpha,>}^{\rm NL}$. Nonlocal magic, instead, is sensitive only to the subleading crossover between the $\nu\simeq 0$ and $\nu\simeq 1$ parts of the spectrum. The resulting logarithmic growth therefore lies beyond the conventional leading-order quasiparticle picture~\cite{alba_2017,alba_2018}, which accounts only for the ballistic contribution of entangled pairs crossing the bipartition. While this does not exclude a more refined quasiparticle description, it shows that the standard pair-counting mechanism is not sufficient to capture the observed behavior.

\paragraph{Conclusions.---}
We a obtained a simple and efficiently computable bound $M_{\alpha,>}^{\rm NL}$ on bipartite nonlocal nonstabilizerness for pure fermionic Gaussian states. We provided an argument for the bound being tight, proved $M_{2,>}^{\rm NL}$ to be a local minimum of nonstabilizerness along the local Gaussian orbit~\cite{supp_material}, and also provided numerical evidence of it being a global minimum in this setting. 

In the Gaussian Haar ensemble, we obtained a closed expression for the average of $M_{2,>}^{\rm NL}$ which, unlike nonlocal magic in Haar-random states~\cite{cao_2025},
is extensive, albeit it contributes only a small fraction of the total 2-SRE. Our results on the Kitaev chain suggests that $M_{2,>}^{\rm NL}$ is suppressed in gapped ground states but exhibits a singularity at quantum phase transitions. In the Kitaev chain, our results are consistent with $M_{2,>}^{\rm NL}\sim \log(L)$ at criticality. For the dynamics of $M_{2,>}^{\rm NL}$, we found a $t^{1/2}$ growth in quantum circuits and linear growth for the XY Hamiltonian. However, in the XX chain, starting from the Neel state, $M_{2,>}^{\rm NL}$ grows only logarithmically, in a striking distinction from the linear growth of entanglement.  

Our results identify free-fermion systems as a rare many-body setting with an easily accessible notion of nonlocal magic. Natural directions this opens for study include the interplay between symmetries and nonlocal magic, e.g., through analogs of symmetry-resolved entanglement~\cite{goldstein2018symmetry,bonsignori2019symmetry} and entanglement asymmetry~\cite{ares2023entanglement}, nonlocal magic in free-fermion topological phases in various dimensions and symmetry classes~\cite{Schnyder2008_PhysRevB.78.195125,kitaev2009periodic,Hasan_Kane_RevModPhys.82.3045,QiZhang_RevModPhys.83.1057}, or in monitored and driven free-fermionic dynamics.

\paragraph{Acknowledgments.---}
M.C. acknowledges support from the Friuli Venezia Giulia Region through the QUASAR-FVG project, funded within the Piano Sviluppo e Coesione FVG.
B.B. was supported by EPSRC grant EP/V062654/1.
E.T. was funded by the Swiss National Science Foundation
(SNSF) under Grant No. TMPFP2\_234754. E.T. acknowledges CINECA (Consorzio Interuniversitario per il Calcolo Automatico) award, under the ISCRA initiative and Leonardo early access program, for the availability of high-performance computing resources and support.

\paragraph{Note Added.---}
While preparing this manuscript, we discovered a closely related independent work by D. Iannotti, B Magni, R. Cioli, A. Hamma, X. Turkeshi, appearing in the same arXiv posting~\cite{toappear}. When overlapping, our results are consistent.

\bibliography{bibliography}

\newpage
\onecolumngrid
\appendix

\section{SUPPLEMENTAL MATERIAL}

\subsection{A matrix inequality}
\label{sec:matrix_inequality}

The following elementary bound is purely algebraic and independent of the fermionic setting.

\begin{lemma}
\label{lem:l4sv}
For any matrix $A\in M_{r\times s}(\RR)$, and integer $\alpha>1$,
$$
\|A\|_{\ell_{2\alpha}}^{2\alpha}\le \sum_k \sigma_k(A)^{2\alpha},
$$
where $\sigma_k(A)$ are the singular values of $A$. Equality holds whenever $A$ is rectangular diagonal.
\end{lemma}

\begin{proof}
One has
$$
\|A\|_{\ell_{2\alpha}}^{2\alpha}=\sum_{j=1}^s\sum_{i=1}^r|A_{ij}|^{2\alpha}
\le
\sum_{j=1}^s\left(\sum_{i=1}^r|A_{ij}|^2\right)^\alpha.
$$
Since $\sum_{i=1}^r|A_{ij}|^2=(A^TA)_{jj}$, it follows that
$$
\|A\|_{\ell_{2\alpha}}^{2\alpha}
\le
\sum_{j=1}^s (A^TA)_{jj}^{\alpha}
\le
\Tr\!\bigl[(A^TA)^\alpha\bigr].
$$
Since $\Tr[(A^TA)^\alpha]=\sum_k \sigma_k(A)^{2\alpha}$, the claim follows. If $A$ is rectangular diagonal, each column has at most one nonzero entry and $A^TA$ is diagonal, so both inequalities are saturated.
\end{proof}

\subsection{Gaussian optimization and the local canonical form}
\label{sec:gaussian_optimization_canonical}

Let $\gamma_1,\dots,\gamma_{2L}$ be Majorana operators, satisfying $\{\gamma_i,\gamma_j\}=2\delta_{ij}$, and adapted to a bipartition $A:B$ with $|A|=2m$, $|B|=2n$, and $m+n=L$. For an ordered subset $I=\{i_1<\cdots<i_p\}\subseteq A$, define the Majorana string $\gamma_I:=\gamma_{i_1}\cdots\gamma_{i_p}$; similarly, for $J=\{j_1<\cdots<j_q\}\subseteq B$, define $\gamma_J:=\gamma_{j_1}\cdots\gamma_{j_q}$. Any state $\rho$ can be expanded in this operator basis as
$$
\rho=\frac{1}{2^L}\sum_{I\subseteq A}\sum_{J\subseteq B} M_{I,J}(\rho)\,\gamma_I\gamma_J,
\qquad
M_{I,J}(\rho):=\Tr(\rho\,\gamma_I\gamma_J),
$$
where the factor $2^{-L}$ follows from the orthogonality of Majorana strings under the Hilbert-Schmidt inner product. For fixed integers $p\in\{0,\dots,2m\}$ and $q\in\{0,\dots,2n\}$, let $M_{p,q}(\rho)$ denote the matrix whose rows are labelled by subsets $I\subseteq A$ with $|I|=p$, whose columns are labelled by subsets $J\subseteq B$ with $|J|=q$, and whose entries are $\bigl(M_{p,q}(\rho)\bigr)_{I,J}:=M_{I,J}(\rho)$.

For a matrix $X=(X_{ab})$, define its entrywise $\ell_r$ norm by $\|X\|_{\ell_r}^r:=\sum_{a,b}|X_{ab}|^r$. The quantity relevant for the Rényi-$\alpha$ problem is then
$$
\mathcal{F}_\alpha(\rho):=\sum_{I,J}|M_{I,J}(\rho)|^{2\alpha}
=
\sum_{p=0}^{2m}\sum_{q=0}^{2n}\|M_{p,q}(\rho)\|_{\ell_{2\alpha}}^{2\alpha}.
$$
Thus, minimizing the corresponding stabiliser Rényi entropy under local unitary transformation is equivalent to maximizing $\mathcal{F}_\alpha$. In particular, for $\alpha=2$ one has $\mathcal{F}_2(\rho)=\sum_{p,q}\|M_{p,q}(\rho)\|_{\ell_4}^4$. Lemma~\ref{lem:l4sv} implies the bound $\|M_{p,q}(\rho)\|_{\ell_\alpha}^\alpha\le \sum_a \sigma_a(M_{p,q}(\rho))^{2\alpha}$, where $\sigma_a(X)$ denotes the singular values of the matrix $X$. 
Therefore, if one finds a local unitary \(U_A \otimes U_B\) such that all matrices \(M_{p,q}\) are diagonalized in their singular-value bases, then the bound is saturated sector by sector and \(\mathcal{F}_\alpha\) is maximized. In general, however, this should be regarded as an exceptional rather than generic situation. The reason is that the singular-value decomposition of \(M_{p,q}\) involves arbitrary orthogonal transformations on the corresponding coefficient spaces, whereas the transformations induced by \(U_A \otimes U_B\) are restricted to the adjoint action of the local unitary groups, \(M_{p,q} \mapsto R_A M_{p,q} R_B^T\), with \(R_A\) and \(R_B\) belonging only to proper subgroups of the relevant orthogonal groups. Hence, although the SVD always exists at the level of abstract orthogonal rotations of the coefficients, there is in general no reason to expect its singular-vector bases to coincide with those that can be reached by physical local unitary rotations of the state. This obstruction becomes even more severe if one further requires the transformation to stay within a restricted class of states, such as the Gaussian manifold.\\

We now specialize to a pure fermionic Gaussian state $\rho_\Gamma$, with real antisymmetric covariance matrix $\Gamma\in M_{2L}(\mathbb{R})$ satisfying $\Gamma^T=-\Gamma$ and $\Gamma^2=-\Id_{2L}$. For an even ordered set $X=\{x_1<\cdots<x_{2r}\}$, Wick's theorem gives $\Tr(\rho_\Gamma\,\gamma_{x_1}\cdots\gamma_{x_{2r}})=i^r\pf(\Gamma_X)$, up to the convention-dependent overall phase. Since only absolute values enter $\mathcal{F}_\alpha$, this phase is immaterial, and one may equivalently work with the Pfaffian matrices
$$
\bigl(M_{p,q}(\Gamma)\bigr)_{I,J}:=\pf(\Gamma_{I\cup J}),
$$
where $\Gamma_{I\cup J}$ is the principal submatrix of $\Gamma$ obtained by restricting rows and columns to the ordered set $I\cup J$. Accordingly, $\mathcal{F}_\alpha(\rho_\Gamma)$ may be evaluated directly from the matrices $M_{p,q}(\Gamma)$.\\

A standard theorem for pure fermionic Gaussian states states that~\cite{botero_2004}, under local orthogonal transformations, $\Gamma$ can be brought to a canonical form. Writing $r:=\min(m,n)$ and assuming for definiteness $m\le n$, there exist $O_A\in O(2m)$ and $O_B\in O(2n)$ such that
$$
\Gamma_{\mathrm{can}}=(O_A\oplus O_B)\Gamma(O_A\oplus O_B)^T
=
\bigoplus_{k=1}^{r}\Gamma_k\oplus \bigoplus_{\ell=1}^{n-r}J_\ell,
$$
up to irrelevant row/column permutations (and in a slightly different convention from the main text), where each entangled block is
$$
\Gamma_k=
\begin{pmatrix}
0 & \nu_k & \mu_k & 0\\
-\nu_k & 0 & 0 & \mu_k\\
-\mu_k & 0 & 0 & -\nu_k\\
0 & -\mu_k & \nu_k & 0
\end{pmatrix},
\qquad
\mu^2_k=1-\nu_k^2,
\qquad
0\le \nu_k\le 1,
$$
and each purely local block on $B$ is $J_\ell=\begin{pmatrix}0&1\\-1&0\end{pmatrix}$. Here $\oplus$ denotes the direct sum of matrices. Each block $\Gamma_k$ couples one fermionic mode in $A$ to one fermionic mode in $B$, while each $J_\ell$ is a local pure block supported entirely on subsystem $B$.

For each entangled block $\Gamma_k$, let $A_k$ and $B_k$ denote the corresponding two-Majorana index sets on the two sides. For $p,q\in\{0,1,2\}$, define $M_{p,q}^{(k)}:=M_{p,q}(\Gamma_k)$. Since $|A_k|=|B_k|=2$, a direct computation gives $M_{0,0}^{(k)}=[1]$, $M_{2,0}^{(k)}=[\nu_k]$, $M_{0,2}^{(k)}=[-\nu_k]$, $M_{1,1}^{(k)}=\mu_k\,\Id_2$, and $M_{2,2}^{(k)}=[-1]$, while all other $M_{p,q}^{(k)}$ vanish. Thus every nonzero elementary matrix is either a scalar or a scalar multiple of the identity. For a local block $J_\ell$ on $B$, the only nonzero coefficient matrices are $M_0(J_\ell)=[1]$ and $M_2(J_\ell)=[-1]$.

Because $\Gamma_{\mathrm{can}}$ is a direct sum of independent blocks, the corresponding Gaussian state factorizes into a tensor product of block states
$$
\rho_{\Gamma_{\mathrm{can}}}
=
\bigotimes_{k=1}^{r}\rho_{\Gamma_k}\otimes \bigotimes_{\ell=1}^{n-r}\rho_{J_\ell}.
$$ 
At fixed total bidegree $(p,q)$, one must sum over all compatible splittings of $p$ and $q$ among the entangled blocks and the local $J_\ell$ blocks. For $m\le n$, this means choosing $(p_k,q_k)\in\{0,1,2\}^2$ for each $\Gamma_k$ and $q_\ell^{\mathrm{loc}}\in\{0,2\}$ for each $J_\ell$, subject to $\sum_k p_k=p$ and $\sum_k q_k+\sum_\ell q_\ell^{\mathrm{loc}}=q$. For each such choice, the corresponding coefficient matrix is the Kronecker product
$$
\bigotimes_{k=1}^{r} M^{(k)}_{p_k,q_k}
\;\otimes\;
\bigotimes_{\ell=1}^{n-r} M_{q_\ell^{\mathrm{loc}}}(J_\ell).
$$
Each nonzero elementary factor is square, hence every such Kronecker product is square as well. This does not contradict the fact that the full matrix $M_{p,q}(\Gamma_{\mathrm{can}})$ is generally rectangular, since the latter is indexed by all subsets $I\subseteq A$ with $|I|=p$ and all subsets $J\subseteq B$ with $|J|=q$, and therefore may also contain rows or columns corresponding to blockwise occupations incompatible with the nonzero elementary sectors. Those rows or columns vanish identically. Up to a permutation of rows and columns, $M_{p,q}(\Gamma_{\mathrm{can}})$ 
is obtained by embedding the direct sum
$$
\bigoplus_{\substack{\mathrm{compatible}\\ \mathrm{splittings}}}
\left(
\bigotimes_{k=1}^{r} M^{(k)}_{p_k,q_k}
\;\otimes\;
\bigotimes_{\ell=1}^{n-r} M_{q_\ell^{\mathrm{loc}}}(J_\ell)
\right)
$$
into the full rectangular matrix of size $\binom{2m}{p}\times \binom{2n}{q}$, with all remaining entries equal to zero. Equivalently, the nonzero support of $M_{p,q}(\Gamma_{\mathrm{can}})$ is the direct sum over compatible splittings, while all rows and columns corresponding to incompatible block-wise occupations vanish identically. Since every nonzero elementary factor is diagonal, or a scalar multiple of the identity, each Kronecker product in the nonzero block is diagonal. 
 
Although the elementary block matrices are scalar or proportional to the identity, the full matrices
\(M_{p,q}(\Gamma_{\rm can})\) are not, in general, diagonal in singular-value bases after summing over all compatible
splittings. Therefore, Lemma~\ref{lem:l4sv} is not saturated sector by sector in general. The following computation should thus be understood as an explicit evaluation of \(F_\alpha\) on the canonical representative, rather than by itself as a proof of optimality.

\subsection{Explicit evaluation of $\mathcal{F}_\alpha$ on the canonical form}

The previous section identifies $\Gamma_{\mathrm{can}}$ as the optimal candidate to be the optimizer for a Gaussian state. Since the coefficient matrices of $\Gamma_{\mathrm{can}}$ factorize block by block, $\mathcal{F}_\alpha$ factorizes as well and can be evaluated explicitly. For a single entangled block $\Gamma_k$, the only nonvanishing elementary matrices are $M_{0,0}^{(k)}=[1]$, $M_{2,0}^{(k)}=[\nu_k]$, $M_{0,2}^{(k)}=[-\nu_k]$, $M_{1,1}^{(k)}=\mu_k\,\Id_2$, and $M_{2,2}^{(k)}=[-1]$, with $\mu_k^2=1-\nu_k^2$. Therefore the total contribution of block $k$ to $\mathcal{F}_\alpha$ is
$$
|1|^{2\alpha}+|\nu_k|^{2\alpha}+|\nu_k|^{2\alpha}+2|\mu_k|^{2\alpha}+|1|^{2\alpha}
=
2\bigl(1+\nu_k^{2\alpha}+\mu_k^{2\alpha}\bigr).
$$
Each local block $J_\ell$ contributes instead $|M_0(J_\ell)|^{2\alpha}+|M_2(J_\ell)|^{2\alpha}=2$.
Hence, defining $r=\min(m,n)$ and $s=\max(m,n)$, one obtains
$$
\mathcal{F}_\alpha(\Gamma_{\mathrm{can}})
= 2^s\prod_{k=1}^{r}\bigl(1+\nu_k^{2\alpha}+\mu_k^{2\alpha}\bigr).
$$
In particular, for $\alpha=2$,
$$
\mathcal{F}_2(\Gamma_{\mathrm{can}})
=
2^{L}\prod_{k=1}^{r}\bigl(1-\nu_k^2+\nu_k^4\bigr),
$$
since $1+\nu_k^4+\mu_k^4=2(1-\nu_k^2+\nu_k^4)$ and $L=r+s$.

\subsection{Local maximality of $\mathcal{F}_\alpha$ at $\Gamma_{\mathrm{can}}$}
\label{subsec:roadmap_local_maximum}

We explain here why the proof of local maximality of $\mathcal{F}_\alpha$ at the canonical covariance matrix $\Gamma_{\mathrm{can}}$ reduces to a two-block computation. Throughout, $\Gamma_{\mathrm{can}}$ is the block-diagonal canonical form already introduced, while the factorized expression of $\mathcal{F}_\alpha(\Gamma_{\mathrm{can}})$ is the one discussed in the previous section. The statement is local, namely it concerns infinitesimal local Gaussian transformations $\Gamma\mapsto (O_A\oplus O_B)\Gamma(O_A\oplus O_B)^T$ around $\Gamma_{\mathrm{can}}$.

The argument has three ingredients.

\paragraph{(1) Intra-block directions.---}
Since $\Gamma_{\mathrm{can}}$ is a direct sum of independent canonical blocks, and $\mathcal{F}_\alpha(\Gamma_{\mathrm{can}})$ factorizes block by block, any local orthogonal transformation acting only inside one canonical block affects only the corresponding single-block factor. 
Concretely, they act as $O_A\oplus O_B$ with $O_{A,B}\in \text{SO}(2)$ thus changing the off-diagonal by an SO($2$) factor.
Thus the intra-block problem decouples completely. Moreover, for each elementary entangled block $\Gamma_a$, the canonical $4\times4$ structure is already the maximizer under arbitrary local orthogonal transformations acting inside that block. Therefore $\mathcal{F}_\alpha$ is already locally maximal along all intra-block directions.

\paragraph{(2) Two-block channels.---}
Fix two entangled blocks, say $\Gamma_{a}$ and $\Gamma_{b}$. On the $A$ side, the infinitesimal local mixings between the two corresponding two-dimensional Majorana subspaces are parametrized by $2\times2=4$ independent generators from the SO($4$)/SO($2$)$\times$SO($2$) action, where we factor out local changes; the same holds on the $B$ side. Hence the most general local infinitesimal mixing of the pair $(a,b)$ depends on $4+4=8$ real parameters. Since all other blocks are spectators, this channel is completely described by the minimal $8\times8$ model made of the two blocks $\Gamma_a\oplus\Gamma_b$. Denoting these eight parameters collectively by $\boldsymbol\theta$, one computes the Hessian $H_\alpha^{(ab)}:=\partial_{\boldsymbol\theta}^2\mathcal{F}_\alpha|_{\boldsymbol\theta=0}$. For $\alpha=2$, this Hessian can be evaluated symbolically in Mathematica, and one finds that all its eigenvalues are negative. Thus $\Gamma_{\mathrm{can}}$ is a strict local maximum in every two-block mixing channel.

\paragraph{(3) Decoupling of the full Hessian into intra-block and pairwise inter-block channels.---}
It remains to show that, for an arbitrary number $r$ of entangled blocks, the full second variation of $\mathcal{F}_\alpha$ around $\Gamma_{\mathrm{can}}$ reduces to the sum of the intra-block contributions and of the individual two-block mixing channels. This is the key structural point.

Let $A=A_1\oplus\cdots\oplus A_r$ and $B=B_1\oplus\cdots\oplus B_r$, where each $A_a$ and $B_a$ is the two-dimensional Majorana subspace associated with the canonical block $\Gamma_a$. Consider a smooth local Gaussian deformation $\Gamma(\varepsilon)=(O_A(\varepsilon)\oplus O_B(\varepsilon))\,\Gamma_{\mathrm{can}}\,(O_A(\varepsilon)\oplus O_B(\varepsilon))^T$, with $O_A(0)=\Id$, $O_B(0)=\Id$, and write $O_A(\varepsilon)=e^{\varepsilon K_A}$, $O_B(\varepsilon)=e^{\varepsilon K_B}$, where $K_A^T=-K_A$ and $K_B^T=-K_B$. The corresponding first-order variation is $\delta\Gamma=\frac{d}{d\varepsilon}\Gamma(\varepsilon)|_{\varepsilon=0}=[K,\Gamma_{\mathrm{can}}]$, with $K=K_A\oplus K_B$. Thus the tangent space at $\Gamma_{\mathrm{can}}$ is the vector space of all matrices of the form $\delta\Gamma=[K,\Gamma_{\mathrm{can}}]$ generated by infinitesimal local Gaussian transformations.

Because $\Gamma_{\mathrm{can}}$ is organized into canonical blocks, each antisymmetric generator decomposes uniquely into block components, namely $K_A=\sum_a K_A^{(a)}+\sum_{a<b}K_A^{(ab)}$ and similarly for $K_B$, where $K_A^{(a)}$ acts only inside the single block $A_a$, while $K_A^{(ab)}$ mixes only the two blocks $A_a$ and $A_b$; the same meaning applies to $K_B^{(a)}$ and $K_B^{(ab)}$ on the $B$ side. We stress that the intra-block pieces need not commute with $\Gamma_{\mathrm{can}}$; they simply generate deformations acting inside one canonical block , and these have already been controlled in paragraph~(1). Accordingly, we split the tangent space into intra-block sectors $W_a$, generated by $(K_A^{(a)},K_B^{(a)})$, and inter-block sectors $V_{ab}$, generated by $(K_A^{(ab)},K_B^{(ab)})$. Since the decomposition of $K_A$ and $K_B$ into block components is unique, the corresponding decomposition of tangent vectors is unique as well. Therefore the inter-block tangent space splits as
the direct sum $\bigoplus_{a=1}^{r}W_a\;\oplus\;\bigoplus_{a<b}V_{ab}$.

We now show that the Hessian $H_\alpha$ of $\mathcal{F}_\alpha$ at $\Gamma_{\mathrm{can}}$ is block diagonal with respect to this decomposition. For each $a$, let $S_a$ be the local sign flip acting as $-\Id$ on the four Majoranas of block $a$ and as $\Id$ on all other blocks. Since $\Gamma_{\mathrm{can}}$ is block diagonal, one has $S_a\Gamma_{\mathrm{can}}S_a^T=\Gamma_{\mathrm{can}}$. On the other hand, $\mathcal{F}_\alpha$ is invariant under $\Gamma\mapsto S_a\Gamma S_a^T$ for every $\Gamma$, because for every principal minor $X$ one has $\pf((S_a\Gamma S_a^T)_X)=\det((S_a)_X)\pf(\Gamma_X)$, and hence $\bigl|\pf((S_a\Gamma S_a^T)_X)\bigr|^{2\alpha}=\bigl|\pf(\Gamma_X)\bigr|^{2\alpha}$. Therefore, for every infinitesimal perturbation $X$, the identity $\mathcal{F}_\alpha(\Gamma_{\mathrm{can}}+\epsilon X)=\mathcal{F}_\alpha(\Gamma_{\mathrm{can}}+\epsilon\,S_aXS_a^T)$ holds. Expanding both sides in $\epsilon$, one concludes that the Hessian is invariant under conjugation by $S_a$, namely
$$
H_\alpha(X,Y)=H_\alpha(S_aXS_a^T,S_aYS_a^T)
$$
for all tangent vectors $X$ and $Y$.

The action of $S_a$ on the tangent sectors is by parity. If $u\in W_a$, then $S_a u S_a^T=u$, because an intra-block generator inside block $a$ picks up two minus signs; if $u\in W_b$ with $b\neq a$, then again $S_a u S_a^T=u$, since block $a$ is not involved. By contrast, if $u\in V_{ab}$, then $S_a u S_a^T=-u$, because a generator mixing blocks $a$ and $b$ picks up only one minus sign from block $a$; if instead $u\in V_{bc}$ with $a\notin\{b,c\}$, then $S_a u S_a^T=u$.

These parity rules imply the vanishing of all mixed Hessian terms between different sectors. First, if $u\in W_a$ and $v\in V_{bc}$ with $a\in\{b,c\}$, then $S_a u S_a^T=u$ while $S_a v S_a^T=-v$, so by invariance of the Hessian one gets $H_\alpha(u,v)=H_\alpha(u,-v)=-H_\alpha(u,v)$, hence $H_\alpha(u,v)=0$. Thus no intra-block/inter-block mixed term survives. Second, if $u\in V_{ab}$ and $v\in V_{cd}$ with $\{a,b\}\neq\{c,d\}$, there exists an index $\ell$ belonging to one pair but not to the other; then $S_\ell u S_\ell^T=-u$ and $S_\ell v S_\ell^T=v$, so again $H_\alpha(u,v)=0$. Finally, if $u\in W_a$ and $v\in W_b$ with $a\neq b$, the mixed Hessian term vanishes because, along purely intra-block deformations, $\mathcal{F}_\alpha(\Gamma_{\mathrm{can}})$ factorizes into independent single-block contributions, so the mixed derivative between two different block coordinates is zero at the stationary point.

We conclude that the Hessian is block diagonal with respect to the decomposition above, and therefore the second variation splits as
$$
\delta^2\mathcal{F}_\alpha=\sum_{a=1}^{r}\delta^2\mathcal{F}_\alpha^{(a)}+\sum_{a<b}\delta^2\mathcal{F}_\alpha^{(ab)}.
$$
The first sum contains the intra-block contributions already controlled in paragraph~(1), while the second contains the pairwise inter-block channels controlled in paragraph~(2). Since each of these terms is non-positive, the full Hessian is negative semidefinite. Moreover, the first variation vanishes because every infinitesimal direction is a sum of intra-block and pairwise inter-block directions, and each of these restricted problems is already extremized at $\Gamma_{\mathrm{can}}$. Hence $\Gamma_{\mathrm{can}}$ is a local maximum of $\mathcal{F}_\alpha$ in the Gaussian manifold.

\subsection{Thermodynamic limit from the Jacobi ensemble}
We start from the finite-size joint probability density of the variables
$x_1,\dots,x_m \in [0,1]$, where $x_k=\nu_k^2$, and $\nu_k$ are the singular values associated with the subsystem block. In the Haar-random pure Gaussian ensemble, these variables are governed by a Jacobi-type joint law of the form
\begin{equation}
P_{m,n}(x_1,\dots,x_m)
=
\frac{1}{Z_{m,n}}
\prod_{1\le i<j\le m}(x_i-x_j)^2
\prod_{k=1}^m x_k^{-1/2}(1-x_k)^{\,n-m},
\qquad x_k\in[0,1],
\end{equation}
where $Z_{m,n}$ is the normalization constant. The one-point density $\rho_{m,n}(x) = \int_{[0,1]^{m-1}} dx_2\cdots dx_m P(x,x_2,\dots,x_m)$ is known to be~\cite{bianchi_2021,jiang_2009}
\begin{equation}
\rho_{m,n}(x) = \frac{w_{m,n}(x)}{m}
\sum_{r=0}^{m-1}\frac{p_r(x)^2}{h_r},
\end{equation}
where $p_r(x) $ are the degree-$r$
Jacobi orthogonal polynomials with weight
$w_{m,n}(x) = x^{-1/2}(1-x)^{n-m-1/2}$
and $h_r$ are their norms. 

However, here we are interested in the thermodynamic regime  $L=m+n \to \infty$ in the scaling limit $\ell=\frac{m}{L} \in [0,1]$ fixed. The goal is to determine the limiting density of the particles $x_k$. To do so, one rewrites the probability density in exponential form $P_{m,n} \equiv \exp(-m^2 \mathcal{E}_{m,n})$. Taking the logarithm of $P_{m,n}$ gives
\begin{equation}
\log P_{m,n}(x_1,\dots,x_m)
=
-\log Z_{m,n}
+
2\sum_{1\le i<j\le m}\log|x_i-x_j|
-\frac12\sum_{k=1}^m \log x_k
+
(n-m)\sum_{k=1}^m \log(1-x_k),
\end{equation}
where the different contributions can be interpreted separately. The term
$
2\sum_{i<j}\log|x_i-x_j|
$
describes the mutual logarithmic repulsion between the particles, while
$
(n-m)\sum_k \log(1-x_k)
$
acts as an external confining potential which pushes the particles away from $x=1$. The term
$
-\frac12\sum_k \log x_k
$
only affects the behavior near $x=0$.

To extract the leading thermodynamic behavior, one must identify the correct scale of the exponent. The repulsion term contains $m(m-1)/2$ pairs, hence it is of order $m^2$. Moreover, $n-m = L-2m = L(1-2\ell)$,
and since $m=\ell L$, the confinement term is also of order $m^2$. 
By contrast, the term involving $\sum_k \log x_k$ contains only $m$ summands and is therefore of order $m$, hence subleading with respect to the $m^2$ contributions. This is why one divides by $m^2$: after this rescaling, the leading terms admit a finite nontrivial limit. At this point one introduces the empirical measure
\begin{equation}
\rho_m(x)=\frac1m\sum_{k=1}^m \delta(x-x_k) \to \rho(x)
\end{equation}
which represents the random density of particle positions,
while $\rho$ denotes the deterministic limiting density as $m\to\infty$.
In terms of $\rho$, sums over particles may be rewritten as integrals (discarding sub-leading correlations in the two-point measure):
\begin{equation}
\frac1m\sum_{k=1}^m g(x_k) \to \int g(x)\rho(x)\,dx,
\quad
\frac{1}{m^2}\sum_{i\neq j} h(x_i,x_j)
\;\longrightarrow\;
\iint h(x,y)\rho(x)\rho(y)\,dx\,dy,
\end{equation}

At finite $m$, the one-point marginal density (introduced in the main text) is defined by
\begin{equation}
\rho_{m,n}(x)
=
\frac{1}{m}\,\mathbb{E}\left[\sum_{k=1}^m \delta(x-x_k)\right] = 
\mathbb{E}\left[ \rho_{m}(x)\right].
\end{equation}
This is the averaged particle density, such that the average that we want to compute is indeed
$\mathbb E[\MNL]/L =  -\ell \,\mathbb{E} \left[\sum_{k=1}^{m}\log(1-x_k+x_k^2)/m\right] = -\ell\int dx \log(1-x+x^2) \rho_{m,n}(x)$.
On the other hand, $\rho_m$ is the empirical density of a single configuration. In the thermodynamic limit, the empirical density concentrates and becomes deterministic, so that $\rho_{m,n}(x)\to \rho(x)$ as well.
Thus, in the large-$L$ limit, the equilibrium density $\rho$ entering the energy functional also coincides with the limiting one-point marginal density.

Replacing the leading sums by their continuum counterparts, one obtains the Coulomb-gas energy functional
\begin{equation}
\frac{\mathcal{E}_{m,n}}{m^2} \to \mathcal{E}[\rho]
=
-\iint \log|x-y|\,\rho(x)\rho(y)\,dx\,dy
-\alpha \int \log(1-x)\,\rho(x)\,dx,
\end{equation}
with $\alpha = (n-m)/m = (1-2\ell)/\ell$.
The first term represents the logarithmic repulsion between particles, and the second term is the contribution of the external potential
$V(x)=-\alpha\log(1-x)$. The probability of observing a macroscopic density $\rho$ behaves, at leading exponential order, like
\begin{equation}
\mathrm{P}[\rho]\asymp e^{-m^2\mathcal E[\rho]}.
\end{equation}
Therefore the thermodynamic density is obtained by minimizing $\mathcal E[\rho]$ under the constraints
$\rho(x)\ge 0$ and $\int \rho(x)\,dx = 1$.
The minimizer satisfies the Euler--Lagrange equation
\begin{equation}
-2\int \log|x-y|\,\rho(y)\,dy - \alpha\log(1-x)=C,
\end{equation}
for $x$ in the support of $\rho$, where $C$ is a constant Lagrange multiplier enforcing normalization. Differentiating with respect to $x$ yields the singular-integral equation
\begin{equation}
2\,\mathrm{PV}\!\int \frac{\rho(y)}{x-y}\,dy
=
\frac{\alpha}{1-x},
\end{equation}
valid on the support of $\rho$. At leading order the support is sought in the form $\mathrm{supp}(\rho)=[0,\tilde x]$, with $\tilde x<1$. This is natural because the potential is logarithmically diverging at $x=1$. To solve the singular integral equation, one introduces the resolvent
\begin{equation}
G(z)=\int_0^{\tilde x} \frac{\rho(y)}{z-y}\,dy,
\qquad
z\in\mathbb C\setminus[0,\tilde x].
\end{equation}
This function is analytic outside the support, and its discontinuity across the cut determines the density:
\begin{equation}
\rho(x)=\frac{1}{2\pi i}\bigl(G_-(x)-G_+(x)\bigr),
\quad 
G_+(x)+G_-(x) = 2\,\mathrm{PV}\!\int \frac{\rho(y)}{x-y}\,dy,
\quad 
G_{\pm}(x) = \lim_{\epsilon\to 0} G(x\pm i \epsilon).
\end{equation}
Moreover, the normalization of $\rho$ implies
$G(z)\sim z^{-1}$ for $z\to\infty$. A one-cut ansatz consistent with the support $[0,\tilde x]$, and satisfying the condition $G_+(x)+G_-(x) = \alpha/(1-x)$ is
\begin{equation}
G(z)
=
\frac{\alpha}{2(1-z)}
-
\frac{A}{1-z}\sqrt{\frac{z-\tilde x}{z}},
\end{equation}
where the square root is chosen so that $\sqrt{(z-\tilde x)/z}\to 1$, for $z\to\infty$. The constants $A$ and $\tilde x$ are fixed by two conditions. First, $G(z)$ must be regular at $z=1$. Hence the apparent pole must cancel,
i.e. $\lim_{z\to 1} (1-z) G(z) = 0$ implying 
\begin{equation}\label{eq:condition1}
\frac{\alpha}{2}-A\sqrt{1-\tilde x}=0.
\end{equation}
Second, the large-$z$ behavior must satisfy $G(z)\sim z^{-1}$. 
Expanding one finds $G(z)=\left(A-\frac{\alpha}{2}\right)\frac1z+O(z^{-2})$, hence
\begin{equation}\label{eq:condition2}
A-\frac{\alpha}{2}=1.
\end{equation}
Solving Eq.s (\ref{eq:condition1},\ref{eq:condition2}) together with the definition $\alpha = (1-2\ell)/\ell$ one gets
\begin{equation}
A=\frac{1}{2\ell}, \qquad \tilde x=4\ell(1-\ell).
\end{equation}
The density is then obtained from the discontinuity of the resolvent. For $x\in(0,\tilde x)$,
\begin{equation}
G_-(x)-G_+(x)
=
\frac{2iA}{1-x}\sqrt{\frac{\tilde x-x}{x}},
\end{equation}
therefore
\begin{equation}\label{eq:single_particle}
\rho(x)
=
\frac{1}{2\pi\ell}\,
\frac{\sqrt{(\tilde x-x)/x}}{1-x}\,
\mathbf{1}_{(0,\,\tilde x)}(x),
\end{equation}
which is the thermodynamic density of the variables $x_k=\nu_k^2$. Equivalently, it is the large-$L$ limit of the one-point marginal density of the Jacobi ensemble. Plugging this result into the definition of the nonlocal magic, for $\alpha=2$ one obtains 
\begin{eqnarray}
\mathcal{J}(\ell) &\equiv & \lim_{L\to\infty} \frac{1}{L} \mathbb{E}_{\Gamma}[M_{2,>}^\text{NL}(\psi_\Gamma)] 
 =  -\frac{1}{2\pi}\int_{0}^{\tilde x} dx
\log(1-x+x^2) \frac{\sqrt{(\tilde x-x)/x}}{1-x} \\
& = & \Re 
\left\{
2\sqrt{1-\tilde x}\log\left[ \frac{\sqrt{1-\tilde x}+\sqrt{1-\tilde x e^{i\pi/3}}}{\sqrt{1-\tilde x}+1}\right]
-
2\log\left[ \frac{1+\sqrt{1-\tilde x e^{i\pi/3}}}{2}\right]
\right\}.
\end{eqnarray}

\end{document}